\newif\ifproceedings\proceedingstrue
\newif\ifsinglecolumn 

\documentclass[conference]{IEEEtran}
\usepackage[hyphens]{url}
\usepackage{graphicx}
\usepackage{multirow}
\usepackage{amsfonts,amssymb}
\usepackage[ruled,algosection,noend,linesnumbered]{algorithm2e}
\usepackage{amsthm,amsmath}
\usepackage{siunitx}
\usepackage{booktabs}
\usepackage{xspace}
\usepackage{xcolor}
\usepackage{framed}
\usepackage[all]{xy}
\usepackage{balance}
\SetMathAlphabet{\mathcal}{normal}{OMS}{lmsy}{m}{n}
\SetMathAlphabet{\mathcal}{bold}{OMS}{lmsy}{m}{n}
\SetSymbolFont{largesymbols}{normal}{OMX}{lmex}{m}{n}
\SetSymbolFont{largesymbols}{bold}{OMX}{lmex}{m}{n}
\usepackage{booktabs}
\usepackage{siunitx}
\usepackage{subcaption}
\usepackage{pifont}

\newtheorem{theorem}{Theorem}[section]
\newtheorem{lemma}[theorem]{Lemma}

\usepackage{microtype}

\newenvironment{prettylist}{
\raggedright
\begin{list}{
    \footnotesize\raisebox{0.1mm}{\small\ding{118}}
}{
    \setlength\topsep{1ex}
    \setlength\leftmargin{32pt}
    \setlength\rightmargin{12pt}
    \setlength\itemsep{2pt}
    \setlength\parskip{0pt}
    \setlength\parsep{0pt}
    \setlength\itemindent{-15pt}
}
}{
\end{list}
}

\newcommand\bi{\begin{itemize}}
\newcommand\ei{\end{itemize}}
\newcommand\ben{\begin{enumerate}}
\newcommand\een{\end{enumerate}}
\newcommand{\trm}[1]{\textrm{#1}}
\newcommand{\tbf}[1]{\textbf{#1}}

\newcommand*\dash{\ifvmode\quitvmode\else\unskip\kern.16667em\fi---%
\hskip.16667em\relax}

\newcommand{\rscoin}{\textsf{RSCoin}\xspace}
\newcommand{\aquery}{\textsf{Query}\xspace}
\newcommand{\acommit}{\textsf{Commit}\xspace}
\newcommand{\aepoch}{\textsf{CloseEpoch}\xspace}

\newcommand{\secrscoin}{RSCOIN\xspace}

\newcommand{\secp}{\lambda}
\newcommand{\usecp}{1^{\secp}}
\newcommand{\randpick}{\xleftarrow{\$}}
\newcommand{\pk}{{\ensuremath{pk}}}
\newcommand{\sk}{{\ensuremath{sk}}}
\newcommand{\sig}{{\ensuremath{\sigma}}}
\newcommand{\inputs}{\mathsf{in}}
\newcommand{\outputs}{\mathsf{out}}

\newcommand{\sigkeygen}{\mathsf{Sig.KeyGen}}
\newcommand{\sigsign}{\mathsf{Sig.Sign}}
\newcommand{\sigverify}{\mathsf{Sig.Verify}}

\newcommand{\checktx}{\mathsf{CheckTx}}
\newcommand{\checknotdoublespent}{\mathsf{CheckNotDoubleSpent}}
\newcommand{\committx}{\mathsf{CommitTx}}

\newcommand{\addr}{\mathsf{addr}}
\newcommand{\addrid}{\mathsf{addrid}}
\newcommand{\block}{\mathsf{b}}
\newcommand{\centralblock}{\mathsf{B}}
\newcommand{\period}{\mathsf{period}}
\newcommand{\tx}{\mathsf{tx}}
\newcommand{\txspec}[3]{\mathsf{tx}({#1}\xrightarrow{{#2}}{#3})}
\newcommand{\epoch}{\mathsf{epoch}}

\newcommand{\bank}{\mathsf{bank}}
\newcommand{\mintette}{\mathsf{m}}

\newcommand{\owners}{\mathsf{owners}}

\newcommand{\bundle}{\mathsf{bundle}}

\newcommand{\seq}{\mathit{seq}}

\newcommand{\periodtxset}{\mathsf{pset}}
\newcommand{\txset}{\mathsf{txset}}
\newcommand{\utxo}{\mathsf{utxo}}
\newcommand{\txindex}[2]{\mathsf{index}_{{#1}}({#2})}
\newcommand{\periodpk}[1]{\addr_\bank^{({#1})}}
\newcommand{\periodmintette}[1]{\mathsf{DPK}_{{#1}}}
\newcommand{\mintettesig}{\sig_\bank^{(\mintette)}}
\newcommand{\mintetteset}{\mathsf{mset}}

\newcommand{\centralblockspec}[2]{\centralblock_{{#1}}^{(#2)}}
\newcommand{\otherblocks}{\mathsf{otherblocks}}

\newcommand{\currencypk}[2]{\pk_{#1}^{(#2)}}

\newcommand{\currencysk}[2]{\pk_{#1}^{(#2)}}
\newcommand{\multiaddr}{\mathsf{multiaddr}}
\newcommand{\spendtx}{\mathsf{SpendTx}}
\newcommand{\refundtx}{\mathsf{RefundTx}}

\title{Centrally Banked Cryptocurrencies}

\ifproceedings{
\author{\IEEEauthorblockN{George Danezis}
\IEEEauthorblockA{University College London\\
\url{g.danezis@ucl.ac.uk}}
\and 
\IEEEauthorblockN{Sarah Meiklejohn}
\IEEEauthorblockA{University College London\\
\url{s.meiklejohn@ucl.ac.uk}}
}
}\fi

\begin{document}

\IEEEoverridecommandlockouts
\makeatletter\def\@IEEEpubidpullup{9\baselineskip}\makeatother
\IEEEpubid{\parbox{\columnwidth}{Permission to freely reproduce all or part
    of this paper for noncommercial purposes is granted provided that
    copies bear this notice and the full citation on the first
    page. Reproduction for commercial purposes is strictly prohibited
    without the prior written consent of the Internet Society, the
    first-named author (for reproduction of an entire paper only), and
    the author's employer if the paper was prepared within the scope
    of employment.  \\
    NDSS '16, 21-24 February 2016, San Diego, CA, USA\\
    Copyright 2016 Internet Society, ISBN 1-891562-41-X\\
    http://dx.doi.org/10.14722/ndss.2016.23187
}
\hspace{\columnsep}\makebox[\columnwidth]{}}

\maketitle

\begin{abstract}

Current cryptocurrencies, starting with Bitcoin, build a decentralized
blockchain-based transaction ledger, maintained through proofs-of-work that
also serve to generate a monetary supply.  Such decentralization has benefits,
such as
independence from national political control, but also significant limitations
in terms of computational costs and scalability.
We introduce \rscoin, a cryptocurrency framework in which
central banks maintain complete
control over the monetary supply, but rely on a distributed set of
authorities, or \emph{mintettes}, to prevent
double-spending.  While monetary policy is centralized, \rscoin still
provides strong transparency and auditability guarantees.
We demonstrate, both theoretically and experimentally, the benefits of
a modest degree of centralization, such as the elimination of wasteful
hashing and a scalable system for avoiding double-spending attacks.
\end{abstract}

\section{Introduction}


Bitcoin~\cite{satoshi-bitcoin}, introduced in 2009, and
the many alternative cryptocurrencies it has inspired (e.g., Litecoin and
Ripple), have achieved enormous success:
financially, in November 2015, Bitcoin held a market capitalization of $4.8$
billion USD and $30$ cryptocurrencies held a market capitalization of
over $1$ million USD.  In terms of visibility, cryptocurrencies have been
accepted as a
form of payment by an increasing number of international merchants, such as
the 150,000 merchants using either Coinbase or Bitpay as a payment gateway
provider.

Recently, major financial institutions such as JPMorgan Chase~\cite{jpmorgan}
and Nasdaq~\cite{nasdaq} have
announced plans to develop blockchain technologies. The potential impacts of
cryptocurrencies have now been
acknowledged even by government institutions: the European Central Bank
anticipates their ``impact on monetary policy and price
stability''~\cite{ecb-on-bitcoin}; the US Federal Reserve
their ability to provide a ``faster, more secure and more efficient payment
system''~\cite{fed-on-bitcoin}; and the UK Treasury vowed to ``support
innovation''~\cite{uk-on-banking} in this space. This is unsurprising, since
the financial settlement systems currently in use by
central banks (e.g., CHAPS, TARGET2, and Fedwire) remain relatively expensive
and\dash at least behind the scenes\dash have high latency and are stagnant
in terms of innovation.

Despite their success, existing cryptocurrencies suffer from a
number of limitations.  Arguably the most troubling one is their poor
scalability: the Bitcoin network (currently by far the most heavily used) can
handle at most $7$ transactions per
second\footnote{\url{http://en.bitcoin.it/wiki/Scalability}} and faces
significant challenges in raising this rate much
higher,\footnote{\url{http://en.bitcoin.it/wiki/Blocksize_debate}} whereas
PayPal handles over $100$ and Visa handles on average anywhere from 2,000 to
7,000.  This lack of scalability is
ultimately due to its reliance on broadcast and the need to expend significant
computational energy in proofs-of-work \dash by
some estimates~\cite[Chapter 5]{bitcoin-textbook}, comparable to the power
consumption of a large power plant\dash in order
to manage the transaction ledger and make double-spending attacks
prohibitively expensive.
Alternative cryptocurrencies such as Litecoin try to distribute this cost,
and Permacoin~\cite{permacoin} tries to repurpose the computation,
but ultimately neither of these solutions eliminates the costs. A second key
limitation of current cryptocurrencies is the loss of control over monetary
supply, providing little to no flexibility for macroeconomic
policy and extreme volatility in their value as currencies.


Against this backdrop, we present \rscoin, a cryptocurrency framework that
decouples the generation of the monetary supply from the maintenance of the
transaction ledger.  Our design decisions were largely motivated by the
desire to create a more scalable cryptocurrency, but were also inspired by the
research agenda of the Bank of England~\cite{boe-agenda},
and the question of ``whether central banks should themselves make use of such
technology to issue digital currencies.'' Indeed, as Bitcoin becomes
increasingly widespread, we expect that this will be a question of interest
to many central banks around the world.

\rscoin's radical shift from traditional cryptocurrencies is to centralize
the monetary supply. Every
unit of a particular currency is created by a particular central bank, making
cryptocurrencies based on \rscoin significantly more palatable to
governments.  Despite this centralization, \rscoin still provides the
benefit over existing (non-crypto) currencies of a transparent transaction
ledger, a distributed system for maintaining it, and a globally visible
monetary supply. This makes monetary policy transparent, allows direct access
to payments and value transfers, supports pseudonymity, and benefits from
innovative uses of blockchains and digital money.


Centralization of the monetary authority also allows \rscoin to address some
of the scalability issues of fully decentralized cryptocurrencies.
In particular, as we describe in
Section~\ref{sec:system}, the central bank delegates
the authority of validating transactions to a number of other institutions
that
we call \emph{mintettes} (following Laurie~\cite{lauriecoin}).  Since
mintettes are\dash unlike traditional cryptocurrency miners\dash known
and may ultimately be held accountable for any misbehavior, \rscoin
supports a simple and fast mechanism for double-spending detection. As
described
in Section~\ref{sec:consensus}, we adapt a variant
of Two-Phase Commit, optimized to ensure the integrity of a transaction
ledger. Thus, we achieve a significantly more scalable system: the modest
experimental testbed that we describe in Section~\ref{sec:impl} (consisting
of only $30$ mintettes running a basic Python implementation of our consensus
mechanism), can process over 2,000 transactions per second, and performance
scales linearly as we increase the number of mintettes.  Most transactions
take less than one second to clear, as compared to many minutes in traditional
cryptocurrency designs.

Beyond scalability, recent issues in the Bitcoin network have demonstrated
that the incentives of miners may be
misaligned,\footnote{https://bitcoin.org/en/alert/2015-07-04-spv-mining} and
recent research suggests that this problem\dash namely, that miners are
incentivized to produce blocks without fully validating all the transactions
they contain\dash is only exacerbated in other
cryptocurrencies~\cite{verifiers-dilemma}.  We therefore discuss
in Section~\ref{sec:role-bank} how mintettes may collect fees for good
service, and how such fees may be withheld from misbehaving or idle mintettes;
our hope is that this framework can lead to a more robust set of incentives.
In a real deployment of \rscoin, we furthermore expect
mintettes to be institutions with an existing relationship to the
central bank, such as commercial banks, and thus to have some existing
incentives to perform this service.

The ultimate goal for \rscoin is to achieve not only a scalable cryptocurrency
that can be deployed and whose supply can be controlled by one central bank,
but a framework that allows \emph{any} central bank to deploy their own
cryptocurrency. In fact, there is interest~\cite{boe-private} to allow other
entities to not only issue instruments that hold value (such as shares and
derivative products), but to furthermore allow some visibility into
transactions
concerning them.  With this in mind, we discuss in
Section~\ref{sec:multiple-banks} what is needed to support some notion of
interoperability between different deployments of \rscoin, how different
currencies can be exchanged in a transparent and auditable way, and how
various considerations\dash such as a pair of central banks that, for either
security or geopolitical reasons, do not support each other\dash can be
resolved without fragmenting the global monetary system.  We also discuss
other extensions and optimizations in Section~\ref{sec:extensions}.

\section{Related Work}

Much of the research on cryptocurrencies either has analyzed the extent to
which existing properties (e.g., anonymity and fairness) are satisfied or has
proposed new methods to improve certain features.  We focus on those works
that are most related to the issues that we aim to address, namely
stability and scalability.

The work on these two topics has been largely attack-based, demonstrating that
even Bitcoin's heavyweight mechanisms do not provide perfect solutions.  As
demonstrated by Eyal and Sirer~\cite{selfish-mining} and Garay et
al.~\cite{EC:GarKiaLeo15}, an attacker can temporarily withhold blocks and
ultimately undermine fairness.  Babaioff et al.~\cite{babaioff} argued that
honest participation in the Bitcoin network was not sufficiently incentivized,
and Johnson et al.~\cite{johnson} and Laszka et al.~\cite{laszka}
demonstrated that in fact some participants might be incentivized to engage in
denial-of-service attacks against each other.  Karame et
al.~\cite{kac:bitcoin:ccs12} and Rosenfeld~\cite{meni:double-spending:2012}
consider how an adversary might take advantage of both mining power and
the network topology to execute a double-spending attack.  Finally, Gervais et
al.~\cite{is-bitcoin-decentralized} looked at the structure of mining pools,
the rise of SPV clients, and the privileged rights of Bitcoin developers and
concluded that Bitcoin was far from achieving full decentralization.  On
the positive side, Kroll et al.~\cite{kroll} analyzed a simplified model of
the
Bitcoin network and concluded that Bitcoin is (at least weakly) stable.

In terms of other constructions, the work perhaps most related to our own
is Laurie's approach of designated authorities~\cite{lauriecoin}.  This
solution, however, does not describe a consensus mechanism or consider a
centralized entity responsible for the generation of a monetary
supply.  The \rscoin framework is also related to the approaches adopted by
Ripple and Stellar, in that the underlying consensus
protocols~\cite{ripple-consensus,stellar-consensus} used by all
three sit somewhere between a fully decentralized setting\dash
in which proof-of-work-based ``Nakamoto consensus''~\cite{bitcoin-sok} has
thus far been adopted almost unilaterally\dash and a
fully centralized setting (in which consensus is trivial).  Within this space,
\rscoin makes different trust assumptions and thus ends up with
different features: both the Stellar and Ripple consensus protocols avoid a
central point of trust, but at the cost of
needing a broadcast channel (because the list of participants is not fixed a
priori) and requiring servers to be in constant direct communication, whereas
our use of a central bank\dash
which, leaving aside any scalability benefits, is ultimately one of the main
goals of this work\dash allows us to avoid both broadcast channels (because
the set of mintettes is known and thus users can contact them directly) and
direct communication between mintettes.

Finally, our approach borrows ideas from a number of industrial solutions.
In particular, our two-layered approach to the blockchain is in part inspired
by the Bitcoin startup Factom, and our consensus mechanism is in part inspired
by Certificate Transparency~\cite{DBLP:journals/cacm/Laurie14}.
In particular, \rscoin, like Certificate Transparency, uses designated
authorities and relies on
transparency and auditability to ensure integrity of a ledger, rather
than full trust in a central party.

\section{Background}\label{sec:background}

\begin{table}[t]
\small
\centering
\begin{tabular}{lccc|c}
\toprule
 & CC & e-cash & Bitcoin & \rscoin\\
\midrule
Double-spending & online & offline & online & online\\
Money generation & C & C & D & C\\
Ledger generation & C & n.a. & D & D*\\
Transparent & no & no & yes & yes\\
Pseudonymous & no & yes & yes & yes\\
\bottomrule
\end{tabular}
\caption{How existing approaches (credit cards, cryptographic e-cash, and
Bitcoin) and how \rscoin compare in terms of the properties they provide.
Double-spending refers to the way the system detects double-spending (i.e.,
as it happens or after the fact).  C stands for centralized, D for
decentralized, and D* for distributed.}
\label{tab:approaches}
\end{table}

In this section, we present a brief background on Bitcoin and traditional
cryptocurrencies, and introduce some relevant notation.  Since \rscoin
adopts properties of other online
payment systems, such as those of credit cards and cryptographic
e-cash, we highlight some of the advantages and disadvantages of each of
these approaches in Table~\ref{tab:approaches}.

\subsection{The Bitcoin protocol}

Bitcoin is a decentralized cryptocurrency introduced in a whitepaper in
2008~\cite{satoshi-bitcoin} and deployed on January 3 2009.  Since then,
Bitcoin has achieved success and
has inspired a number of alternative cryptocurrencies (often dubbed
``altcoins'') that are largely based on the same blockchain technology.
The novelty of this blockchain technology is that it fulfills the two key
requirements of a currency\dash the generation of a monetary supply and the
establishment of a transaction ledger\dash in a completely decentralized
manner: a global peer-to-peer network serves both to generate new units of
currency and to bear witness to the transfer of existing units from one party
to another through transaction broadcast and computational proof-of-work
protocols.

To highlight the differences between \rscoin and fully decentralized
cryptocurrencies such as Bitcoin, we sketch the main operations and entities
of these blockchain-based currencies; for a more comprehensive overview, we
refer the reader to Bonneau et al.~\cite{bitcoin-sok}.  Briefly, users can
generate signing keypairs and use the public key as a
\emph{pseudonym} or \emph{address} in which to store some units of the
underlying cryptocurrency. To transfer the value stored in
this address to the address of another user, he creates a
\emph{transaction}, which is cryptographically signed using the secret key
associated with this address. More generally, transactions can transfer
value from $m$ input addresses to $n$ output addresses, in which case the
transaction must be signed by the secret keys associated with each of
the input addresses.

Once a user has created a transaction, it is broadcast to his peers in the
network, and eventually reaches \emph{miners}.  A miner seals the transaction
into the global ledger by including it in a pool of transactions, which she
then
hashes\dash along with some metadata and, crucially, a nonce\dash to attempt
to produce a hash below a target value (defined by the difficulty of the
network).  Once a miner is successful in producing such a hash, she broadcasts
the pool of transactions and its associated hash as a \emph{block}.  Among the
metadata for a block is a reference to the previously mined block, allowing
the
acceptance of the miner's block into the \emph{blockchain} to be signaled by
the broadcast of another block with a reference to hers (or, in practice, many
subsequent blocks).  Miners are incentivized by two rewards: the collection
of optional fees in individual transactions, and a
system-specific mining reward (e.g., as of November 2015,
Bitcoin's mining reward of 25~BTC).  These rewards are collected in a special
\emph{coin generation} transaction that the miner includes in her
block's pool of transactions.
Crucially, blocks serve to not only generate the monetary supply (via the
mining rewards included in each block), but also to provide a partial ordering
for transactions: transactions in one block come before transactions included
in any block further along the blockchain.  This allows all users in
the network to eventually impose a global (partial) ordering on transactions,
and thus thwart double-spending by maintaining a list of
\emph{unspent transaction outputs}
and validating a transaction only if its input addresses appear in this list.

What we have described above is the typical way of explaining Bitcoin at
a high level, but we mention that in reality, bitcoins are not ``stored'' in
an address or ``sent''; instead, the sender relinquishes control
by broadcasting a transaction that re-assigns to the recipient's address
the bitcoins previously associated with that of the sender.  An input to a
transaction is thus not an address but a (signed) script that specifies an
index in a previous transaction in which some bitcoins were received; this
\emph{address identifier} uniquely identifies one particular usage of an
address, which becomes important as addresses are reused.  In what
follows, we thus frequently use the notation for an address and for a
transaction-index pair interchangeably.

\subsection{Notation}\label{sec:notation}

We denote a hash function as $H(\cdot)$ and a signature scheme as the tuple
$(\sigkeygen,\sigsign,\sigverify)$, where these algorithms behave as follows:
via $(\pk,\sk)\randpick\sigkeygen(\usecp)$ one generates a signing keypair;
via $\sig\randpick\sigsign(\sk,m)$ one generates a signature; and via
$0/1\gets\sigverify(\pk,m,\sig)$ one verifies a signature on a message.

We use $\addr$ to denote an address; this is identical to a
public key $\pk$ in terms of the underlying technology,\footnote{Or, as in
    Bitcoin, it may be some hashed version of the public key.}
but we use the separate term to disambiguate between usage in a transaction
(where we use $\addr$) and usage as a signature verification key (where we use
$\pk$).
We use $\txspec{\{\addr_i\}_i}{n}{\{\addr_j\}_j}$ to denote a transaction in
which $n$ units of currency are sent from $\{\addr_i\}_i$ to $\{\addr_j\}_j$.
Each usage of an address $\addr$ can be uniquely identified by the tuple
$\addrid=(\tx,\txindex{\tx}{\addr},v)$, where $\tx$ is the hash of the
transaction in which it received some value $v$, and $\txindex{\tx}{\addr}$ is
the index of $\addr$ in the list of outputs.  When
we use these \emph{address identifiers} later on, we occasionally omit
information (e.g., the value $v$) if it is already implicit or unnecessary.
%

\section{An Overview of \secrscoin}\label{sec:overview}

In this section, we provide a brief overview of \rscoin, which will be useful
for understanding both its underlying consensus algorithm (presented in
Section~\ref{sec:consensus}) and the composition of the system as a whole
(presented in Section~\ref{sec:system}).

At a high level, \rscoin introduces a degree of centralization into the two
typically decentralized components of a blockchain-based ledger: the
generation of the monetary supply and the constitution of the transaction
ledger.
In its simplest form, the \rscoin system assumes two structural entities: the
\emph{central bank}, a centralized entity that ultimately has complete control
over the generation of the monetary supply, and a distributed set of
\emph{mintettes} (following Laurie~\cite{lauriecoin}) that are responsible for
the maintenance of the transaction ledger.  The interplay between these
entities\dash and an overview of \rscoin as a whole\dash can be seen in
Figure~\ref{fig:structure}.

\begin{figure}[t]
\centering
\ifsinglecolumn{
\includegraphics[width=0.5\textwidth]{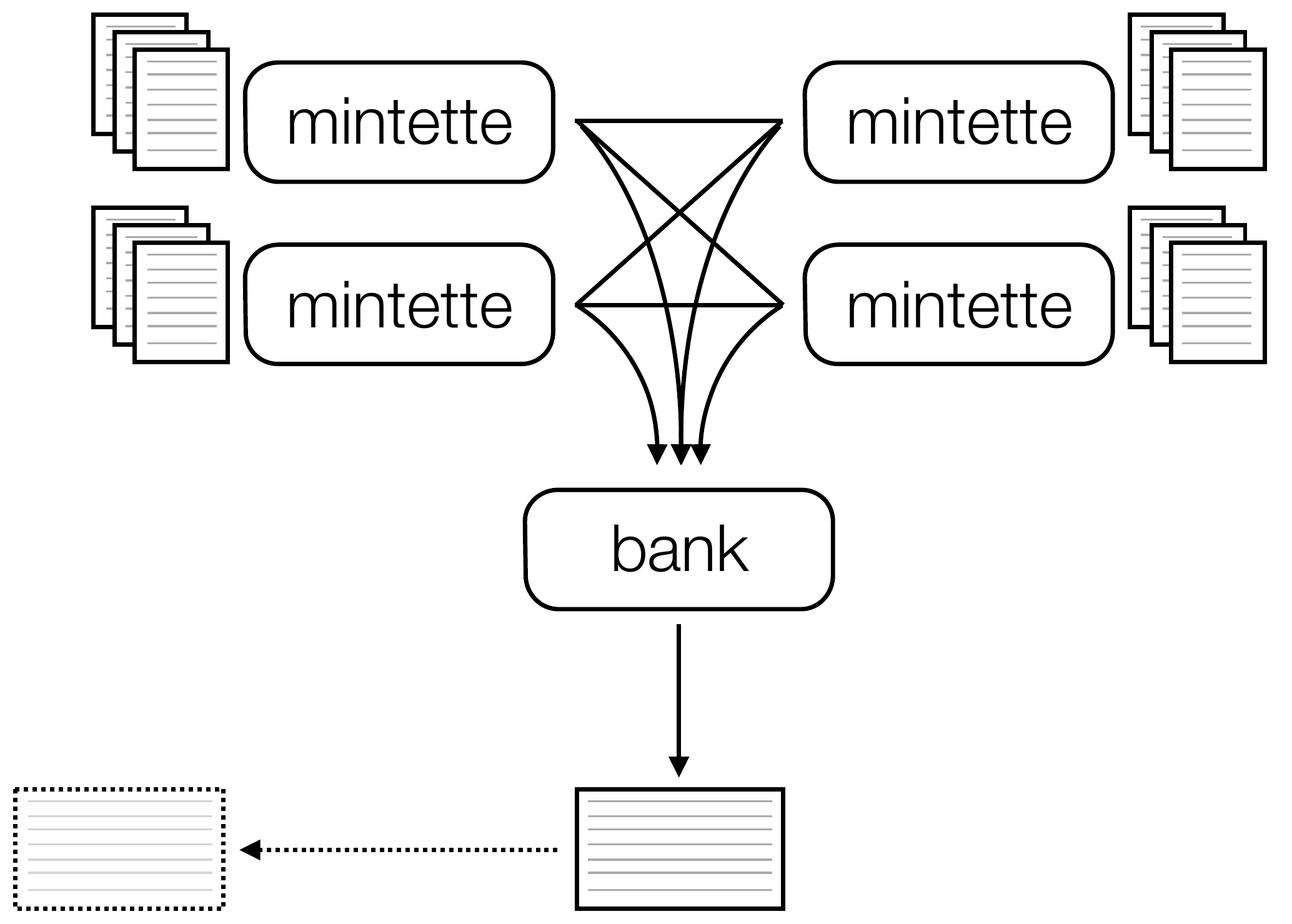}
}\else{
\includegraphics[width=\linewidth]{structure.pdf}
}\fi
\caption{The overall structure of \rscoin.  Each mintettes maintains a set of
lower-level blocks, and (possibly) communicates with other mintettes (either
directly or indirectly).  At some point, the mintettes send these blocks to
the central bank, which produces a higher-level block.  It is these
higher-level blocks that form a chain and that are visible to external users.}
\label{fig:structure}
\end{figure}

Briefly, mintettes collect transactions from users and collate them into
blocks, much as is done with traditional cryptocurrencies.  These mintettes
differ from traditional cryptocurrency miners, however, in a crucial way:
rather than performing some computationally difficult task, each mintette is
simply authorized by the central bank to collect transactions.  In \rscoin,
this authorization is accomplished by a PKI-type functionality, meaning the
central bank signs the public key of the mintette, and each lower-level block
must contain one of these signatures in order to be considered valid.
We refer to the time interval in which blocks are produced by mintettes as an
\emph{epoch}, where the length of an epoch varies depending on the mintette.
Because these blocks are not ultimately incorporated into the main blockchain,
we refer to them as \emph{lower-level blocks}.  Mintettes are collectively
responsible for producing a consistent ledger, and thus to facilitate this
process they communicate internally throughout the course of an epoch\dash
in an indirect manner described in Section~\ref{sec:consensus}\dash and
ultimately reference not only their own previous blocks but also the previous
blocks of each other.  This means that these lower-level blocks form a
(potentially) \emph{cross-referenced} chain.

At the end of some longer pre-defined time interval called a \emph{period},
the mintettes present their blocks to the central bank, which merges these
lower-level blocks to form a consistent history in the form of a new block.
This \emph{higher-level block} is what is ultimately incorporated into the
main blockchain, meaning a user of \rscoin need only keep track of
higher-level blocks. (Special users wishing to audit the behavior of the
mintettes and the central bank, however, may keep track of lower-level blocks,
and we describe in Section~\ref{sec:auditability} ways to augment lower-level
blocks to improve auditability.)

Interaction with \rscoin can thus be quite
similar to interaction with existing cryptocurrencies, as the structure of its
blockchain is nearly identical, and users can create new pseudonyms and
transactions in the same way as before.  In fact, we stress that \rscoin is
intended as a framework rather than a stand-alone cryptocurrency, so one could
imagine incorporated techniques from various existing cryptocurrencies in
order to achieve various goals.  For example, to ensure privacy for
transactions, one could adapt existing cryptographic techniques such as those
employed by Zerocoin~\cite{DBLP:conf/sp/MiersG0R13},
Zerocash~\cite{DBLP:conf/sp/Ben-SassonCG0MTV14}, Pinocchio
Coin~\cite{DBLP:conf/ccs/DanezisFKP13}, or Groth and
Kohlweiss~\cite{DBLP:conf/eurocrypt/GrothK15}.  As these goals are somewhat
orthogonal to the goals of this paper, we leave a comprehensive exploration of
how privacy-enhancing and other techniques can be combined with \rscoin as an
interesting avenue for future work.

\section{Achieving Consensus}\label{sec:consensus}

\begin{figure}[t]
\centering
\ifsinglecolumn{
\includegraphics[width=0.6\textwidth]{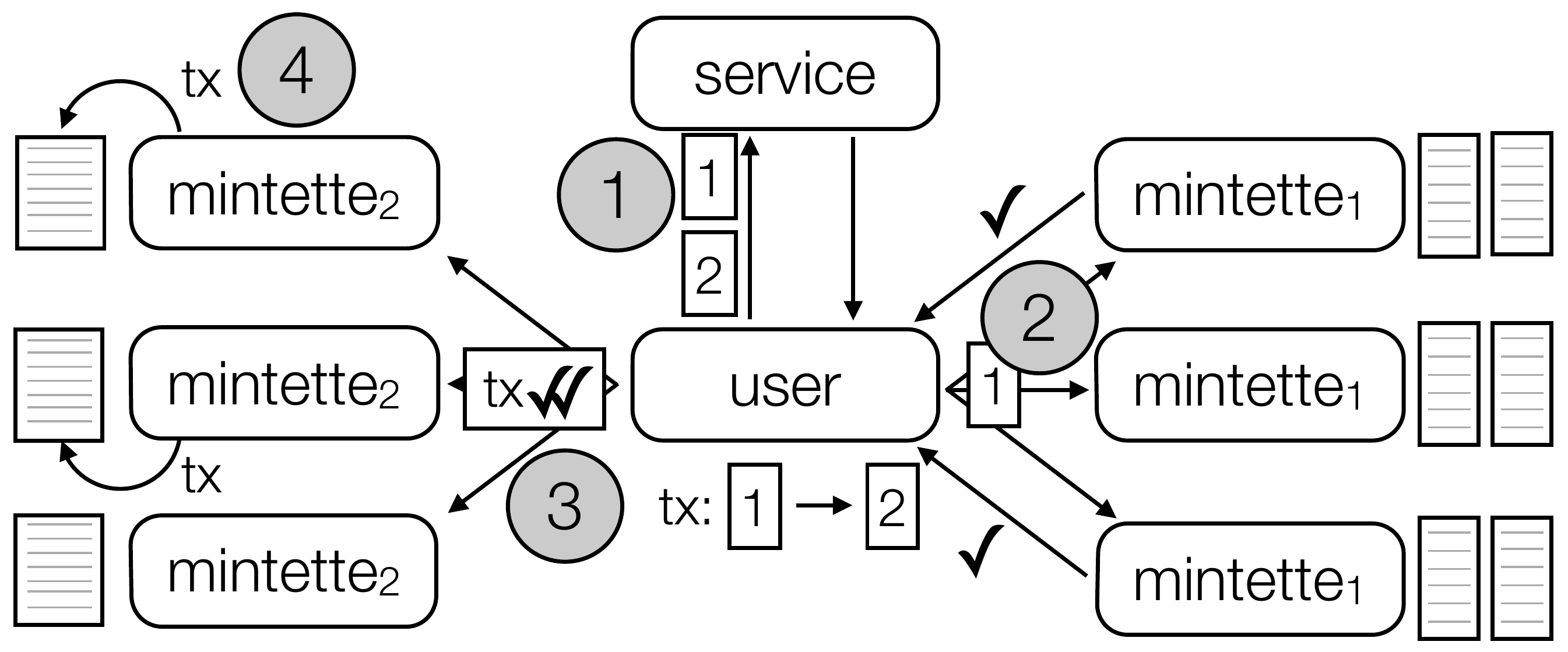}
}\else{
\includegraphics[width=\linewidth]{txvalidation.pdf}
}\fi
\caption{The proposed protocol for validating transactions; each mintette
$\mintette_i$ is an owner of address $i$.  In (1), a user learns the owners of
each of the addresses in its transaction.  In (2), the user collects approval
from a majority of the owners of the input addresses.  In (3), the user sends
the transaction and these approvals to the owners of the transaction
identifier.  In (4), some subset of these mintettes add the transaction to
their blocks.}
\label{fig:txvalidation}
\end{figure}

In the previous section, we described how mintettes send so-called
``lower-level blocks'' to the central bank at the end of a period.  In this
section, we describe a consensus protocol by which these blocks can already
be made consistent when they are sent to the central bank, thus ensuring that
the overall system remains scalable by allowing the central bank to do
the minimal work necessary.

As described in the introduction, one of the major benefits of centralization
is that, although the generation of the transaction ledger is still
distributed, consensus on valid transactions can be reached in a way that
avoids the wasteful proofs-of-work required by existing cryptocurrencies.
In traditional cryptocurrencies, the set of miners is neither known
nor trusted, meaning one has no choice but to broadcast a transaction to the
entire network and rely on proof-of-work to defend against Sybil attacks.
Since our mintettes are in fact authorized by the central
bank, and thus both known and\dash because of their accountability\dash
trusted to some extent, we can avoid the heavyweight consensus requirement of
more fully decentralized cryptocurrencies and instead use an adapted version
of Two-Phase Commit (2PC), as presented in Figure~\ref{fig:txvalidation}. A
generic consensus protocol, ensuring total ordering of
transactions, is not necessary for double-spending prevention; instead, a
weaker property\dash namely that any transaction output features as a
transaction input in at most one other transaction\dash is sufficient.
\rscoin builds its consensus protocol for double-spending prevention based on
this insight.

We begin by describing a threat model for the consensus protocol before going
on to present a basic protocol that achieves consensus on transactions
(Section~\ref{sec:basic-consensus}), an augmented protocol that allows for
auditability of both the mintettes and the central bank
(Section~\ref{sec:auditability}), and a performance evaluation
(Section~\ref{sec:performance}).

\subsection{Threat model and security properties}\label{sec:security}

We always assume that the central bank is
honest\dash although we describe in Section~\ref{sec:role-bank} ways to detect
certain types of misbehavior on the part of the bank\dash and that the
underlying cryptography is secure; i.e., no parties
may violate the standard properties offered by the hash function and
digital signature.  Honest mintettes follow the protocols as specified,
whereas
dishonest mintettes may behave arbitrarily; i.e., they may deviate from the
prescribed protocols, and selectively or broadly ignore requests from
users.  Finally, honest users create only valid transactions (i.e., ones in
which they own the input addresses and have not yet spent their contents),
whereas dishonest users may try to double-spend or otherwise subvert the
integrity of \rscoin.

We consider two threat models.  Our first threat model assumes that each
transaction is processed by a set of mintettes with an honest majority; this
is different from assuming that a majority of all mintettes are honest, as we
will see in our description of transaction processing in
Section~\ref{sec:basic-consensus}.
Our second threat model assumes that no mintette is honest, and that
mintettes may further collude to violate the integrity of \rscoin. This
is a very hostile setting, but we show that some security
properties still hold for honest users.  Additionally, we show that mintettes
that misbehave in certain ways can be detected and ultimately held
accountable, which may serve as an incentive to follow the protocols
correctly. 

In the face of these different adversarial settings, we try to satisfy at
least some of the following key integrity properties:

\begin{prettylist}
\item{\tbf{No double-spending}:} Each output address of a valid transaction
will
only ever be associated with the input of at most one other valid transaction.

\item{\tbf{Non-repudiable sealing}:} The confirmation that a user receives
from a mintette\dash which promises that a transaction will be included in
the ledger\dash can be used to implicate that mintette if the
transaction does not appear in the next block.

\item{\tbf{Timed personal audits}:} A user can, given access to the
lower-level
blocks produced within a period, ensure that the implied behavior of a
mintette matches the behavior observed at the time of any previous
interactions
with that mintette.

\item{\tbf{Universal audits}:} Anyone with access to the lower-level blocks
produced within a period can audit all transactions processed by all
mintettes. In particular, mintettes cannot retroactively modify, omit, or
insert transactions in the ledger.

\item{\tbf{Exposed inactivity}:} Anyone with access to the lower-level blocks
produced within a period can observe any mintette's substantial absence from
participation in the 2PC protocol. (In particular, then, a mintette cannot
retroactively act to claim transaction fees for services not provided in a
timely manner.)
\end{prettylist}

To see how to satisfy these security properties, we first present our basic
consensus protocol in Section~\ref{sec:basic-consensus}, and then present in
Section~\ref{sec:auditability} ways to augment this protocol to achieve
auditability.  We then prove that at least some subset of these security
properties can be captured in both our threat models, and that exposure may
disincentive mintettes from violating those that we cannot capture directly.

\subsection{A basic consensus protocol}\label{sec:basic-consensus}

To begin, the space of possible transaction identifiers is divided
so that each mintette $\mintette$ is responsible for some subset,
or ``shard.''  For reliability and security, each shard is covered by
(potentially) multiple mintettes, and everyone is aware of the owner(s) of
each.

We use $\owners(\addrid)$ to denote the set of mintettes responsible for
$\addrid$.  Recall that $\addrid=(\tx,i,v)$, where $\tx$ specifies the
transaction in which $\addr$, at sequential output $i$, received value $v$. We
map each $\addrid$ to a shard using $\tx$
by hashing
a canonical representation of the transaction.  As a result, all input
$\addrid$ in a transaction
may have different owners (because the addresses may have appeared as an
output
in different transactions), but all output $\addrid$ have the same owner
(because they are all appearing as an output in the same transaction).  For
simplicity, we therefore use the notation $\owners(S_\outputs)$ below
(where $S_\outputs$ is the list of output addresses for a transaction).

In each period, each mintette $\mintette$ is responsible for maintaining two
lists concerning only the $\addrid$ (and indirectly the transactions $\tx$) it
owns: a list of unspent transaction outputs, denoted $\utxo$, and two lists of
transactions seen thus far in the period, denoted $\periodtxset$ and $\txset$
respectively (the former is used to detect double-spending, and the latter is
used to seal transactions into the ledger).  The $\utxo$
list is of the form $\addrid\mapsto(\addr,v)$, where
$(\addrid\mapsto(\addr,v))\in\utxo$ indicates that $\addrid$ had not acted as
an
input address at the start of the period but has since sent value $v$ to
$\addr$
and $(\addrid\mapsto(\bot,\bot))\in\utxo$ indicates that $\addrid$ has not yet
spent its
contents.  The $\periodtxset$ list is of the form $\addrid\mapsto\tx$, where
$(\addrid\mapsto\tx)\in\periodtxset$ indicates that $\addrid$ has acted as an
input address in transaction $\tx$.  We assume
that each mintette starts the period with an accurate $\utxo$ list (i.e.,
all transactions within the mintette's shard in which the outputs have not yet
been spent) and with an empty $\periodtxset$.

At some point in the period, a user creates a transaction.  The
user\footnote{We refer to the user here and in the sequel, but in
    practice this can all be done by the underlying client, without any
    need for input from the (human) user.}
can now run Algorithm~\ref{alg:validate-tx}.\footnote{All algorithms are
    assumed to be executed atomically and sequentially by each party, although
    as we demonstrate in Section~\ref{sec:impl}, implementing them
    using optimistic locking is possible to increase parallelism and
    efficiency.}

\begin{algorithm}
\SetKwComment{Comment}{//}{}
\DontPrintSemicolon
\KwIn{a transaction $\txspec{S_\inputs}{n}{S_\outputs}$ and period
identifier $j$}
$\bundle\gets\emptyset$\;
\Comment{first phase: collect votes}
\ForAll{$\addrid\in S_\inputs$}{
$M\gets \owners(\addrid)$\;

\ForAll{$\mintette\in M$}{
$(\pk_\mintette,\sig)\gets\checknotdoublespent(\tx,\addrid,\mintette)$\;
    \label{alg:validate:vote}
\If{$(\pk_\mintette,\sig)=\bot$}{
\KwRet $\bot$\;
}
\Else
{
$\bundle\gets \bundle\cup\{((\mintette,\addrid)\mapsto(\pk_\mintette,\sig))\}$
}
}
}
\Comment{second phase: commit}
$M\gets\owners(S_\outputs)$\;
\ForAll{$\mintette\in M$}{
$(\pk_\mintette,\sig)\gets\committx(\tx,j,\bundle,\mintette)$
    \label{alg:validate:commit}
}
\caption{Validating a transaction, run by a user}
\label{alg:validate-tx}
\end{algorithm}

In the first phase, the user asks the relevant mintettes to ``vote'' on the
transaction; i.e., to decide if its input addresses have not already been
used, and thus certify that no double-spending is taking place.  To do this,
the user
determines the owners for each input address, and sends the
transaction information to these mintettes, who each run
Algorithm~\ref{alg:checknotdoublespent}.
We omit for simplicity the formal description of an algorithm $\checktx$
that, on input a transaction, checks that the basic structure of the
transaction is valid; i.e., that the collective
input value is at least equal to the collective output value, that the input
address identifiers point to valid previous transactions, and that the
signatures
authorizing previous transaction outputs to be spent are valid.

\begin{algorithm}
\SetKwComment{Comment}{//}{}
\DontPrintSemicolon
\KwIn{a transaction $\tx_c$, an address identifier $\addrid=(\tx,i)$ and
a mintette identifier $\mintette$}
\If{$\checktx(\tx_c) = 0$ {\bf or} $\mintette \notin \owners(\addrid)$ }{
\KwRet $\bot$
}
\Else{
\If{$(\addrid\in\utxo_\mintette)$ {\bf or}
$((\addrid\mapsto\tx_c)\in\periodtxset_\mintette)$}{
$\utxo_\mintette\gets\utxo_\mintette\setminus\{\addrid\}$\;
$\periodtxset_\mintette\gets\periodtxset_\mintette\cup\{(\addrid\mapsto\tx_c)\}$\;
\label{alg:vote:update}
\KwRet $(\pk_\mintette,\sigsign(\sk_\mintette,(\tx_c,\addrid)))$
    \label{alg:vote:send}
} \Else {
\KwRet $\bot$
}
}
\caption{$\checknotdoublespent$, run by a mintette}
\label{alg:checknotdoublespent}
\end{algorithm}

Briefly, in Algorithm~\ref{alg:checknotdoublespent} the mintette first checks
if the current transaction is valid and if
the address is within its remit, and returns $\bot$ otherwise.   It
then proceeds if the address identifier either has not been spent
before (and thus is in $\utxo$), or if it has already
been associated with the given transaction (and thus the pair is in
$\periodtxset$). In those cases, it removes the
address identifier from $\utxo$ and associates it with the transaction
in $\periodtxset$; these actions are
idempotent and can be safely performed more than once. The mintette then
returns a signed acknowledgment to the user.  If instead another transaction
appears in $\periodtxset$ associated with the address identifier, then
the address is acting as an input in two different transactions\dash i.e.,
it is double-spending\dash and the mintette returns $\bot$.  It may also store
the two transactions to provide evidence of double spending.

At the end of the first phase, an honest user will have received some
signatures (representing `yes' votes) from the owners of the input addresses
of the new transaction.  Users should check the signatures
returned by these mintettes and immediately return a failure if any is
invalid.
Once the user has received signatures from at least a
majority of owners for each input, she can now send the transaction,
coupled with a ``bundle of evidence'' (consisting of the signatures of the
input mintettes) to represent its validity, to the owners of the
output addresses (who, recall, are the same for all output addresses). These
mintettes then run Algorithm~\ref{alg:committx}.

\begin{algorithm}[t]
\SetKwComment{Comment}{//}{}
\DontPrintSemicolon
\KwIn{a transaction $\txspec{S_\inputs}{n}{S_\outputs}$, a period identifier
$j$, a bundle of evidence $\bundle=\{((\mintette_i,\addrid_i)
\mapsto(\pk_i,\sig_i))\}_i$, and a mintette identifier $\mintette$}
\If{$\checktx(\tx) = 0$ {\bf or} $\mintette \notin \owners(S_\outputs)$}{
\KwRet $\bot$
}
\Else{
$d\gets 1$\;
\ForAll{$\addrid\in S_\inputs$}{
\ForAll{$\mintette'\in \owners(\addrid)$}{
\If{$(\mintette',\addrid)\in\bundle$}{
$(\pk,\sig)\gets\bundle[(\mintette', \addrid)]$\;
$d'\gets d\land H(\pk)\in\periodmintette{j}\break
~~~~~~~~~~\land \sigverify(\pk,(\tx,\addrid),\sig)$\;
} \Else {
$d\gets 0$
}
}
}
\If{$d = 0$}{
\KwRet $\bot$
} \Else {
$\utxo_\mintette\gets \utxo_\mintette\cup S_\outputs$\;
    \label{alg:commit:utxo}
$\txset_\mintette\gets \txset_\mintette\cup \{\tx\}$\;
    \label{alg:commit:txset}
\KwRet $(\pk_\mintette,\sigsign(\sk_\mintette,\tx))$
    \label{alg:commit:confirmation}
}
}
\caption{$\committx$, run by a mintette}
\label{alg:committx}
\end{algorithm}

In Algorithm~\ref{alg:committx}, a mintette first checks the transaction
and whether it falls within its remit. The mintette then checks the bundle of
evidence by verifying that all\dash or, in practice, at least a majority\dash
of
mintettes associated
with each input are all included, that the input mintettes were authorized to
act as mintettes in the current period, and that their signatures verify.  If
these checks pass and the transaction has not been seen before, then the
mintette adds all the output addresses for the transaction to its $\utxo$
list and adds the transaction to $\txset$. The mintette then sends to the user
evidence that the transaction will be included in the higher-level block
(which a user may later use to implicate the mintette if this is not the
case).

At the end of the period, all mintettes send $\txset$ to the central bank,
along with additional information in order to achieve integrity, which we
discuss in the next section.

\paragraph{Security} In our first threat model, where all transactions are
processed by a set of mintettes with honest majority, it is clear that (1) no
double-spending transactions will be accepted into $\txset$ by honest
mintettes, and (2) the confirmation given to a user in
Line~\ref{alg:commit:confirmation} of Algorithm~\ref{alg:committx} can be
wielded by the user as evidence that the mintette promised to seal the
transaction.  Thus, in our first threat model\dash in which all transactions
are processed by a set of mintettes with honest majority\dash the first and
second integrity properties in Section~\ref{sec:security} are already
satisfied by our basic consensus protocol.

\paragraph{Communication overhead} Importantly, all communication between the
mintettes is done \emph{indirectly} via the user (using the bundles of
evidence), and there is no direct communication between them.
This allows for a low communication overhead for the mintettes, especially
with respect to existing systems such as Bitcoin and Ripple/Stellar (in which
the respective miners and servers must be in constant communication),
which facilitates\dash as we will see in Section~\ref{sec:impl}\dash the
scalability and overall performance benefits of \rscoin.

\subsection{Achieving auditability}\label{sec:auditability}

While our basic consensus mechanism already achieves some of our desired
integrity properties (at least in our weaker threat model), it is still not
clear that it provides any stronger notions of integrity, or that it provides
any integrity in a more hostile environment.  To address this limitation, we
present in this section a way to augment both the lower-level blocks discussed
in Section~\ref{sec:low-blocks} and the basic consensus mechanism.  At a high
level, a mintette now maintains a high-integrity log that highlights both its
own key actions, as well as the actions of those mintettes with whom it has
indirectly interacted (i.e., from whom it has received signatures, ferried
through the user, in the process of committing a transaction).

In more detail, each mintette maintains a log of absolutely ordered actions
along with their notional sequence number. Actions may have one of three
types: \aquery, \acommit and \aepoch. The \aquery action signals an update to
$\periodtxset$ as a result of an input address being assigned to a new
transaction (Line~\ref{alg:vote:update} of
Algorithm~\ref{alg:checknotdoublespent}), so for this action the log includes
the new transaction.  The \acommit action signals an update to $\utxo$ and
$\txset$ as a result of receiving a new valid transaction
(lines~\ref{alg:commit:utxo} and~\ref{alg:commit:txset} of
Algorithm~\ref{alg:committx}, respectively), so for this action the log
includes the transaction and its corresponding bundle of evidence.

To facilitate the \aepoch action, each mintette stores not only the log itself
but also a rolling hash chain; i.e., a \emph{head} that acts as a witness to
the current state of the log, so $h_\seq = H(a_\seq \| h_{\seq-1})$,
where $a_\seq$ is the log entry of the action and $h_{\seq-1}$ is the previous
head of the chain.

To share this witness, mintettes include a
signed head in every message they emit; i.e., in line~\ref{alg:vote:send}
of Algorithm~\ref{alg:checknotdoublespent} and
line~\ref{alg:commit:confirmation} of Algorithm~\ref{alg:committx}, the
mintette $\mintette$ computes
$\sig\randpick\sigsign(\sk_\mintette,(\tx_c,\addrid,h,\seq)$ (where $h$ is
the head of its chain) rather than
$\sig\randpick\sigsign(\sk_\mintette,(\tx_c,\addrid))$, and outputs
$(\pk_\mintette,\sig,h,\seq)$.
Now that mintettes are potentially aware of each others' logs, the
\aepoch action\dash which, appropriately, marks the end of an
epoch\dash includes in the log the heads of the other chains of which the
mintette is aware, along with their sequence number.  This results in the
head of each mintette's chain depending on the latest known head of both its
own and other chains; we refer to this phenomenon as \emph{cross-hashing}
(which, in effect, implements a cryptographic variant of vector
clocks~\cite{DBLP:journals/computer/RaynalS96}).

We can now argue that these augmented lower-level blocks provide sufficient
insight into the actions of the mintettes that stronger notions of integrity
can be achieved.  In particular, we have the following lemma:

\begin{lemma}\label{lem:security}
In both of our threat models, the augmented consensus protocol outlined above
provides timed personal audits, universal audits, and exposed inactivity (as
defined in Section~\ref{sec:security}).
\end{lemma}

\begin{proof}
(Informal.)
To prove that our protocol provides timed personal audits, observe that if
the log reported by any
mintette (or equivalently its hash at any log position) forks at any point
from the record of a user or other mintette, then
the signed head of the hash chain serves as evidence that the log
is different.  To remain undetected, the mintette must therefore provide
users with the signed head of a hash chain that is a prefix of the actual
hash chain it will report.
Both the \aquery and \acommit messages leading to a signed hash, however,
modify the action log.  Providing an outdated hash thus would not contain
the latest action, so again there is evidence that such an action should have
been recorded (in the form of the signed response to the message that should
prompt the action), which also incriminates the mintette. Thus a mintette that
does not wish to be detected and incriminated may only refrain from responding
to requests requiring actions that would change its log.

To prove that our protocol provides universal audits and exposed inactivity,
we first note that,
despite the lack of synchronization between mintettes within periods, we can
detect when an action is committed to a mintette log a `significant time'
after another action.  This is due to the fact that the second message of the
2PC protocol that users send to mintettes carries the hash heads from
all input mintettes involved.  This forms a low-degree random graph with good
expansion properties, and we expect that in a short amount of time mintettes
will have hash chains dependent on the hash chains of all other mintettes.
Thus, if two actions are separated by a sufficiently long period of time, it
is extremely likely that a head dependent on the first action has propagated
to a
super-majority of other mintettes. Checking this property allows us to detect
which came first with very high probability. Using this observation, everyone
may audit claims that a mintette contributed to an action (e.g., processing
the first query of the 2PC protocol for a valid transaction) in a timely
fashion, by using the process above to detect whether the claimed action
from the mintette is or is not very likely to have come after the same action
was committed by all other mintettes concerned.
\end{proof}

Finally, \rscoin makes the key security assumption that all shards are
composed of an honest majority of mintettes.  This is
not quite the same as assuming an overall honest majority of mintettes, but it
can be related to the more traditional assumption that each
mintette behaves honest with some probability, as we demonstrate in the
following lemma:

\begin{lemma}\label{lem:threshold}
Given a fraction of $\alpha$ corrupt mintettes, the probability that $y$
shards,
composed each of $Q$ mintettes, all have an honest majority is
\begin{equation*}
\Pr[\text{secure}] = F\left( \frac{Q-1}{2}; Q; \alpha \right)^y,
\end{equation*}
where $F(k;N;p)$ is the cumulative distribution function of a binomial
distribution over a population of size $N$ with a probability of success $p$.
\end{lemma}

\begin{proof}
The probability that a single shard composed from random mintettes has an
honest
majority is directly the cumulative distribution $\rho = F\left(
\frac{Q-1}{2}; Q; \alpha \right)$.
Since security requires an honest majority across \emph{all} shards we get
$\Pr[\text{secure}] =  \rho^y$.
\end{proof}

This lemma demonstrates that the higher the number of shards, the lower the
probability that all of them will be secure (i.e., covered by an honest
majority of mintettes).  Thus, we recommend fixing the number of shards, on
the
basis of load balancing requirements, to the smallest practical number.  A
mapping can then be defined between the address space and the shards by simply
partitioning equally the space of address identifiers amongst them. For a
given
total number of mintettes $M$, the minimal number of shards of size $Q$ that
should be used is $\lfloor{M / Q}\rfloor$.

\subsection{Performance}\label{sec:performance}

\subsubsection{Theoretical analysis}\label{sec:envelope-numbers}

Looking back at the algorithms in Section~\ref{sec:basic-consensus}, we
can get at least a theoretical estimate
of the communication and computational complexity of the system.  Denote by
$T$ the set of transactions that are generated per
second; by $Q$ the number of mintettes that own each address; and by $M$
the number of total mintettes.

For a transaction with $m$ inputs and $n$ outputs, a user sends and receives
at most $mQ$ messages in the first phase of the 2PC protocol
(line~\ref{alg:validate:vote} of Algorithm~\ref{alg:validate-tx}) and sends
and receives at most $Q$ messages in the second phase
(line~\ref{alg:validate:commit}).  For the user, each transaction thus
requires at most $2(m+1)Q$ messages.

In terms of the communication complexity per mintette, we assume that each
mintette receives a proportional share of the total transactions, which is
ensured as the volume of transactions grow, by the bank allocating shards of
equal sizes to all mintettes. Then the work per mintette is
\[
\frac{\sum_{\tx\in T} 2(m_{\tx}+1)Q}{M}.
\]
In particular, this scales \emph{infinitely}: as more mintettes are added to
the system, the work per mintette decreases (in a linear fashion) and
eventually goes to zero.

\subsubsection{Experimental analysis}\label{sec:impl}

To verify these performance estimates and to measure the latency a typical
user would experience to confirm a transaction, we implemented the basic
consensus mechanism presented in Section~\ref{sec:basic-consensus} and
measured its performance on a modest cluster hosted on Amazon's Elastic
Compute (EC2) infrastructure.
Our implementation\footnote{Available at
\url{https://github.com/gdanezis/rscoin}} consists of 2458 lines of Python
code: 1109 lines define
the core transaction structure, cryptographic processing,
and 2PC protocols as a Twisted service and client; 780 lines are devoted to
unit and timing tests; and 569 lines use the Fabric framework to do
configuration, deployment management (DevOps), live testing, and
visualizations.
Both the clients and the mintettes are implemented as single-threaded services
following a reactor pattern.  All cryptographic operations use the OpenSSL
wrapper library \verb#petlib#, and we instantiate the hash function and
digital
signature using SHA-256 and ECDSA (over the NIST-P224 curve, as optimized by
K{\"a}sper~\cite{DBLP:conf/fc/Kasper11}) respectively.  The implementation and
all configuration
and orchestration files necessary for replicating our results are available
under a BSD license.

Our experimental setup consisted of 30 mintettes, each running on an Amazon
EC2
{\tt t2.micro} instance in the EU (Ireland) data center (for reference, each
cost \$0.014 per hour as of August 2015).  We assigned three mintettes to each
shard of the transaction space, so a quorum of at least two was required for
the 2PC.  A different set of 25
servers on the same data center was used for stress testing and estimating
the peak throughput in terms of transactions per second. Each of those test
machines issued 1000 transactions consisting of two inputs and two outputs.
For wide area networking latency experiments we used a residential broadband
cable service and an Ubuntu 14.02.2 LTS Linux VM running on a 64-bit Windows
7 laptop with a \SI{2.4}{GHz} i7-4700MQ processor and \si{16}{GB} RAM.

\begin{table}
\small
\centering
\begin{tabular}{lS[table-format=7.2,group-separator={,}]S[table-format=4.2,group-separator={,}]}
    \toprule
    {\bf Benchmark} & {$\mu$ ($\si{\second}^{-1}$)} & {$\sigma$} \\ \midrule
    {Hash} & 1017384.86 & 41054.93 \\
    {Sign} & 17043.63 & 2316.40 \\
    {Verify} & 4651.20 & 89.84 \\ \midrule
    {Check tx} & 3585.02 & 95.17 \\
    {Query msg} & 1358.31 & 120.20 \\
    {Commit msg} & 1006.49 & 31.66 \\
    \bottomrule
\end{tabular}
\caption{Micro-benchmarks at the mintettes}
\label{tab:micro}
\end{table}

Table~\ref{tab:micro} reports the mean rate and the standard deviation of
key operations we rely on for \rscoin.\footnote{All measurements were
    performed on a single thread on a single core, using a reactor pattern
    where networking was necessary.}
\emph{Hash}, \emph{Sign} and \emph{Verify} benchmark the number of basic
cryptographic operations each mintette can perform per second (including the
overhead of our library and Python runtime).

For the other benchmarks, we consider a single transaction with one input and
two outputs (we observe that as of September 2014, 53\% of Bitcoin
transactions
had this structure, so this is a reasonable proxy for real usage).
The \emph{check tx} benchmark then measures the rate at which
a mintette can parse and perform the cryptographic checks associated with this
transaction. This involves a single signature check, and thus its difference
from the \emph{Sign} benchmark largely represents the overhead of parsing and
of binary conversion in Python. Guided by this benchmark, we chose to
represent
ECDSA public keys using uncompressed coordinates due to orders-of-magnitude
slowdowns when parsing keys in compressed form.

The \emph{query msg} and \emph{commit msg} benchmarks measure the
rate at which each mintette can process the first and second message of the
2PC respectively for this transaction.
These include full de-serialization, checks from persistent storage of the
$\utxo$, cryptographic checks, updates to the $\utxo$, signing, and
serialization of responses. These benchmarks guided our design towards not
synchronizing to persistent storage the $\utxo$ before each response, and
relying instead on the quorum of mintettes to ensure correctness (a design
philosophy similar to
RAMCloud~\cite{DBLP:journals/cacm/OusterhoutAEKLMMNOPRRSS11}). Persisting to
storage before
responding to each request slowed these rates by orders of magnitude.

\begin{figure}[t!]
\centering
\begin{subfigure}[b]{0.9\linewidth}
\centering
\ifsinglecolumn{
\includegraphics[width=0.45\textwidth]{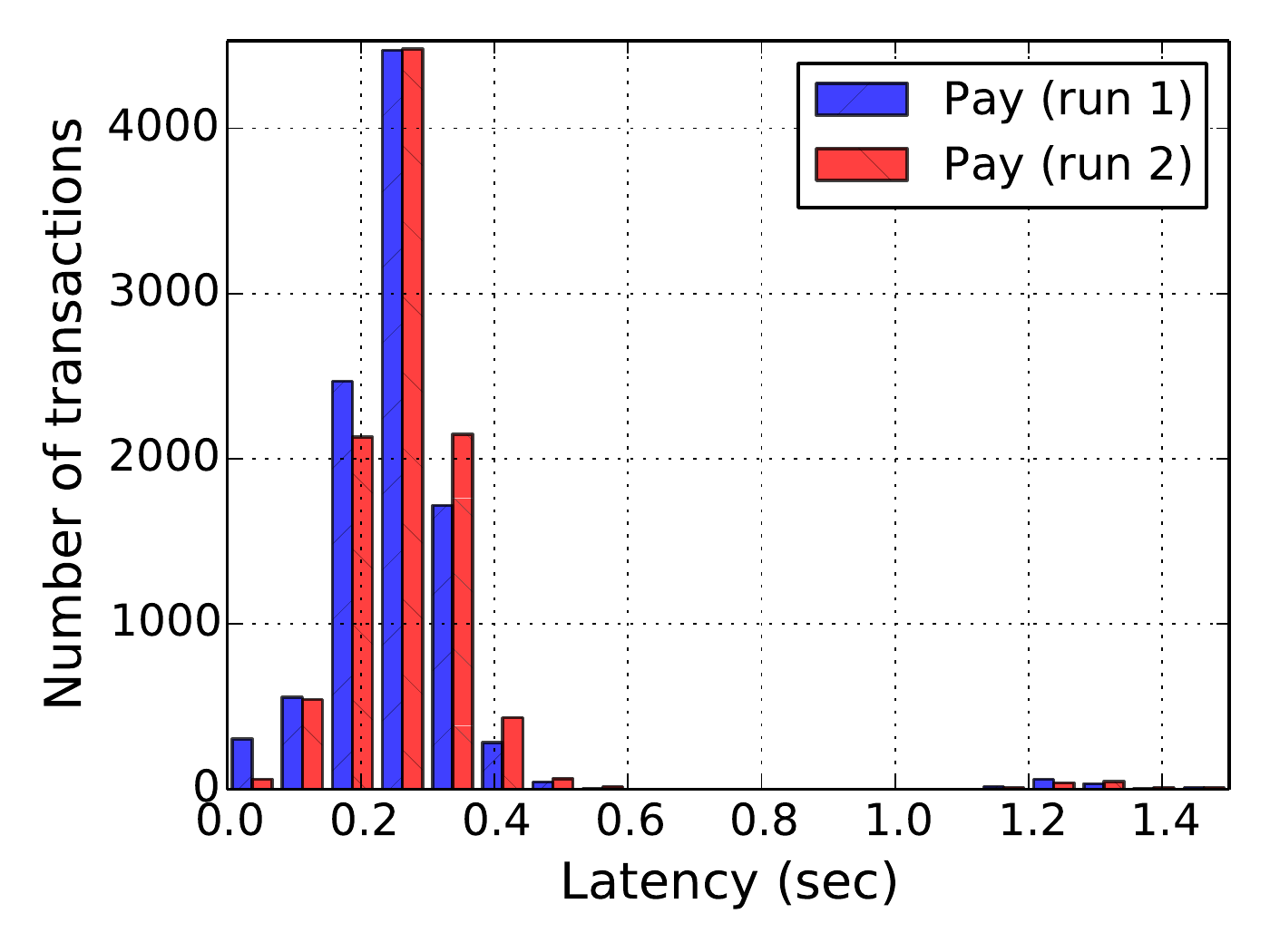}
}\else{
\includegraphics[width=0.8\linewidth]{WAN-latency.pdf}
}\fi
\caption{Local area network (EC2)}
\label{fig:lanlatency}
\end{subfigure}
~
\begin{subfigure}[b]{0.9\linewidth}
\centering
\ifsinglecolumn{
\includegraphics[width=0.45\textwidth]{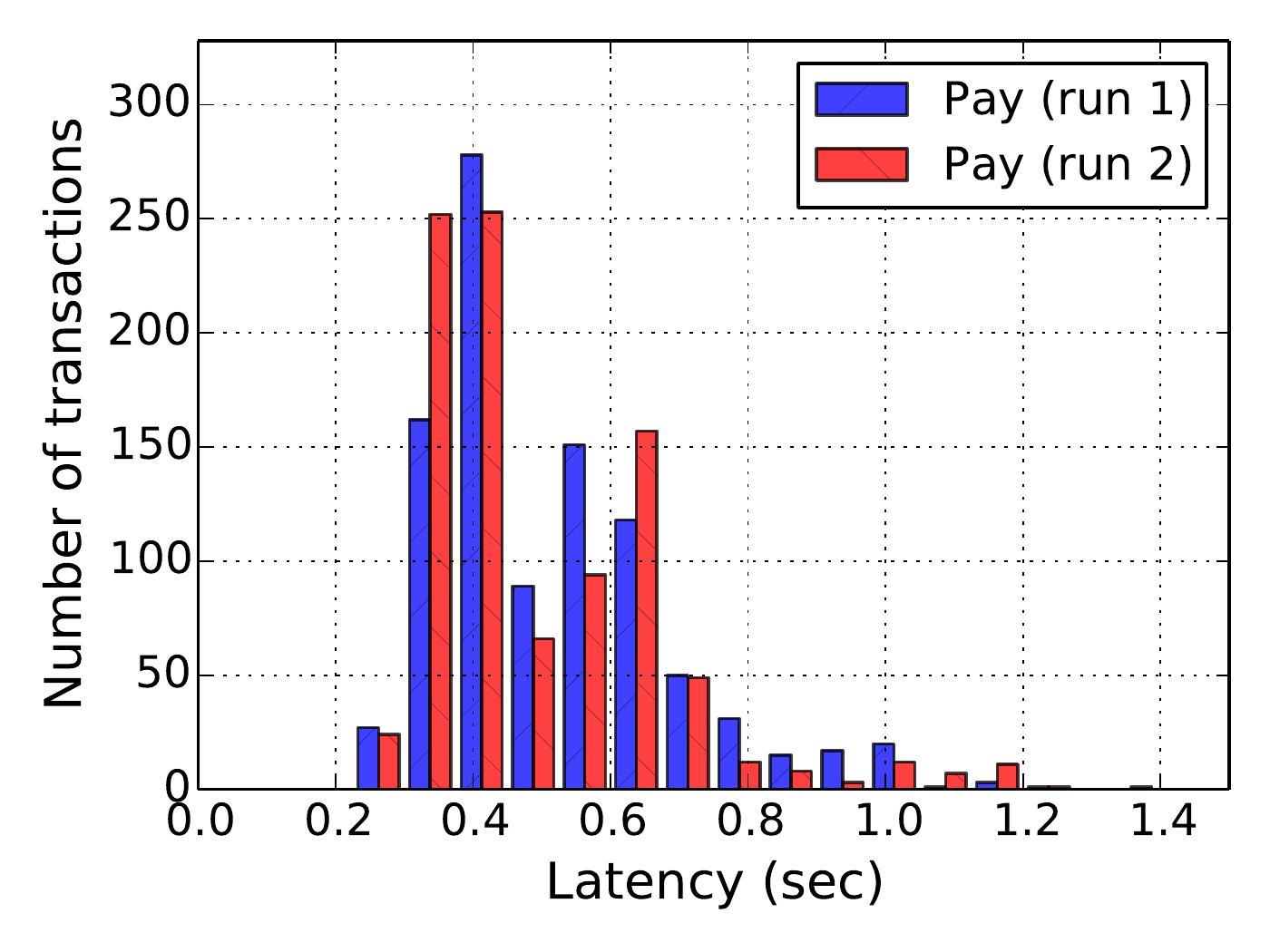}
}\else{
\includegraphics[width=0.8\linewidth]{LAN-latency.pdf}
}\fi
\caption{Wide area network (Broadband)}
\label{fig:wanlatency}
\end{subfigure}
\caption{Latency, in seconds, to perform the 2PC to validate
a payment for a transaction with freshly issued coins as inputs (run 1), and
transactions with two arbitrary previous transactions as inputs (run 2).}
\label{fig:latency}
\end{figure}

Figure~\ref{fig:latency} illustrates the latency a client would experience
when interacting with the mintettes. Figure~\ref{fig:lanlatency}
illustrates the experiments with client machines within the data center, and
point to an intrinsic delay due to networking overheads and cryptographic
checks of less than 0.5 seconds. This includes both phases of the 2PC.

Over a wide area network the latency increases (Figure~\ref{fig:wanlatency}),
but under the conditions tested, the latency
is still usually well under a second for the full 2PC and all checks.  We
note that no shortcuts were implemented: for each transaction, all three
mintettes for each input were contacted and expected to respond in the first
phase, and all three mintettes responsible for the new transaction were
contacted and have to respond in the second phase. In reality, only a majority
need to respond before concluding each phase, and this may reduce latency
further.

\begin{figure}[t!]
\centering
\ifsinglecolumn{
\includegraphics[width=0.5\textwidth]{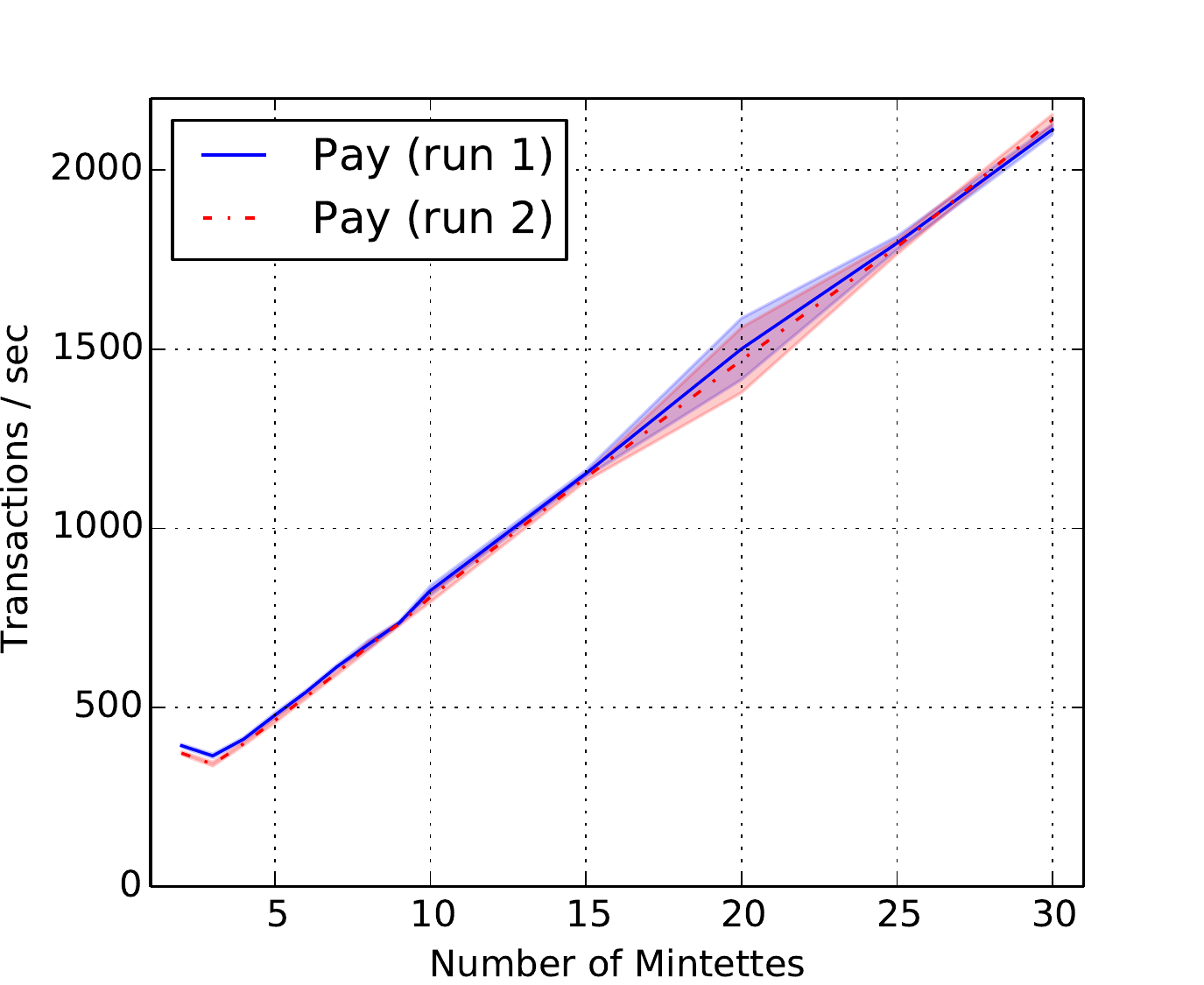}
}\else{
\includegraphics[width=0.9\linewidth]{Throughput-60.pdf}
}\fi
\caption{Throughput (90$^{\text{th}}$ percentile and standard error), in
transactions per second, as a function of the number of mintettes, for
transactions with two freshly issued coins as inputs (run 1) and transactions
with two arbitrary previous transactions as inputs (run 2).}
\label{fig:throughput}
\end{figure}

Figure~\ref{fig:throughput} plots the throughput of the system as we increase
the number of mintettes from 2 to 30, under the load of 25 synthetic clients,
each pushing 1000 transactions.  As expected, when fewer than three mintettes
are available the throughput is roughly flat (fewer than 400 transactions per
second), as both phases of the 2PC need to contact all mintettes.  Once more
than the minimum of three mintettes are available the load is distributed
across them: the first phase need to access at most six mintettes (three
for each of the two transaction inputs), and the second phase at most three
mintettes. This load per transaction is independent of the number of mintettes
and as a result the throughput scales linearly, as predicted in
Section~\ref{sec:envelope-numbers}.  After the initial three mintettes, each
new mintette adds approximately 66 additional transactions per second to the
capacity of the system.

The gap between the micro-benchmarks relating to the message processing for
the two phases (1358.31~$\si{\second}^{-1}$ and 1006.49 $\si{\second}^{-1}$
respectively) and the rate of transactions observed under end-to-end
conditions (approximately 400 $\si{\second}^{-1}$) indicates that at this
point bandwidth, networking, or the interconnection with the process are
scaling bottlenecks for single mintettes. In particular no pipelining was
implemented as part of the client (although the mintettes support it) and thus
every request initiates a fresh TCP connection, with the slowdowns and
resource
consumption on the hosts that this entails.

\section{The \secrscoin System}\label{sec:system}

With our consensus protocol in place, we now describe the structure of
\rscoin,
focusing on the interaction between the mintettes and the central bank, and on
the overall parameters and properties of the system.  We first describe the
structure and usage of \rscoin
(Sections~\ref{sec:low-blocks} and~\ref{sec:high-blocks}) and then
address considerations that arise in how to allocate fees to mintettes
(Section~\ref{sec:role-bank}); overlay \rscoin on top of an existing
cryptocurrency like Bitcoin (Section~\ref{sec:embedding}); incentivize
mintettes to follow the consensus protocol and present a collectively
consistent ledger to the central bank
(Section~\ref{sec:incentives}); and set concrete choices for various
system parameters (Section~\ref{sec:parameters}).

\subsection{Lower-level blocks}\label{sec:low-blocks}

A lower-level block produced by a mintette $\mintette$ within $\period_i$
looks like $\block = (h,\txset,\sig,\mintetteset)$,
where $h$ is a hash, $\txset$ is a collection of transactions, and
$\sig$ is a signature from the mintette that produced this block.  The fourth
component $\mintetteset$ specifies the cross-chain property of lower-level
blocks (recall from Section~\ref{sec:auditability} that mintettes may
reference each others' blocks) by identifying the hashes of the other previous
blocks that are being referenced.

Denote by $\pk_\bank$ the bank's public key and by $\periodmintette{i}$ the
set of mintettes authorized by the bank in the previous higher-level block
$\centralblockspec{\bank}{i-1}$
(as described in Section~\ref{sec:high-blocks}), and define
$\otherblocks \gets h_1\|\ldots\| h_n$ for $\mintetteset = (h_1,\ldots,h_n)$.
Assuming the block $\block$ is produced in $\epoch_j$, to check that $\block$
is valid one then checks that
\ben
\item $h = H(h_{\bank}^{(i-1)}\|h_{j-1}^{(\mintette)}
    \|\otherblocks\|\txset)$,
\item $\sigverify(\pk_\mintette,h,\sig)=1$,
\item $(\pk_\mintette,\mintettesig)\in\periodmintette{i}$ for some
$\mintettesig$, and
\item $\sigverify(\pk_\bank,(\pk_\mintette,\period_i),\mintettesig)=1$.
\een

To form a lower-level block, a mintette uses the transaction set $\txset$ it
has formed throughout the epoch (as described in
Section~\ref{sec:basic-consensus}) and the hashes $(h_1,\ldots,h_n)$ that it
has
received from other mintettes (as ferried through the ``bundle of evidence''
described in Section~\ref{sec:auditability})
and creates $\mintetteset\gets(h_1,\ldots,h_n)$,
$\otherblocks\gets h_1\|\ldots\| h_n$,
$h\gets H(h_{\bank}^{(i-1)}
\| h_{j-1}^{(\mintette)}\|\otherblocks\|\txset)$,
and $\sig\randpick\sigsign(\sk_\mintette,h)$.

\subsection{Higher-level blocks}\label{sec:high-blocks}

The higher-level block that marks the end of $\period_i$ looks like
$\centralblockspec{\bank}{i} = (h,\txset,\allowbreak\sig,
\periodmintette{i+1})$, where these first three
values are similar to their counterparts in lower-level blocks (i.e., a hash,
a collection of transactions, and a signature), and the set
$\periodmintette{i+1}$
contains pairs $(\pk_\mintette,\mintettesig)$; i.e., the public
keys of the mintettes authorized for $\period_{i+1}$ and the bank's
signatures on the keys.

To check that a block is valid, one checks that
\ben
\item $h = H(h_{\bank}^{(i-1)}\|\txset)$,
\item $\sigverify(\pk_\bank,h,\sig)=1$, and
\item $\sigverify(\pk_\bank,(\pk_\mintette,\period_{i+1}),\mintettesig)=1$ for
all $(\pk_\mintette,\mintettesig)\in\periodmintette{i+1}$.
\een

To form a higher-level block, the bank must collate the inputs it is given by
the mintettes, which consist of the lower-level blocks described above and the
action logs described in Section~\ref{sec:auditability}.  To create a
consistent transaction set $\txset$, a vigilant bank might need to look
through all of the transaction sets it receives to detect double-spending,
remove any conflicting transactions, and identify the mintette(s) responsible
for including them.  As this would require the bank to perform work
proportional
to the number of transactions (and thus somewhat obviate the reason for
mintettes), we
also consider an optimistic approach in which the bank relies on the consensus
protocol in Section~\ref{sec:consensus} and instead simply merges the
individual transaction sets to form $\txset$.  The bank then forms $h\gets
H(h_{\bank}^{(i-1)}\|\txset)$, $\sig\randpick\sigsign(\pk_\bank,h)$, and
creates the set of
authorized mintettes using a decision process we briefly discuss below and in
Section~\ref{sec:incentives}.

\subsubsection{Coin generation and fee allocation}\label{sec:role-bank}

In addition to this basic structure, each higher-level block could also
contain within $\txset$ a special coin generation transaction and an
allocation
of fees to the mintettes that earned them in the previous period.
Semantically, the coin generation could take on the same structure as in
Bitcoin; i.e., it could be a transaction $\txspec{\emptyset}{n}{\addr_\bank}$,
where $\addr_\bank$ is an address owned by the bank, and fees could be
allocated using a transaction $\txspec{\addr_\bank}{f}{\addr_\mintette}$,
where
$f$ represents the fees owed to $\mintette$.
The interesting question is thus not how central banks can allocate fees to
mintettes, but how it decides which mintettes have earned these fees.  In
fact, the provided action logs allow the central bank to identify active and
live mintettes and allocate fees to them appropriately.

This mechanism (roughly) works as follows.  The central bank keeps a tally
of the mintettes that were involved in certifying the validity of input
addresses; i.e., those that replied in the first phase of the consensus
protocol.  The choice to
reward input mintettes is deliberate: in addition to providing a direct
incentive for mintettes to respond in the first phase of the protocol, it
also provides an indirect incentive for mintettes to respond in the second
phase, as only a transaction output that is marked as unspent can later be
used as an input (for which the mintette can then earn fees).  Thus,
rewarding input mintettes provides incentive to handle a transaction
throughout its lifetime.

The action logs also play a crucial role in fee allocation.  In particular,
the
``exposed inactivity'' security property from Section~\ref{sec:auditability}
prevents an inactive mintette from
becoming active at a later time and claiming that it contributed to previous
transactions, as an examination of the action logs can falsify such claims.
Additionally, if fee allocation is determined on the basis of a known function
of the action logs, anyone with access to the action logs can audit the
actions of the central bank.

Finally, we mention that although the logs are sent only to the central bank,
the expectation is that the central bank will publish these logs to allow
anyone to audit the system, as well as the bank's operation.  As we assume the
central bank is honest, this
does not present a problem, but in a stronger threat model in which less trust
were placed in the central bank, one might instead attempt to adopt a
broadcast system for distributing logs (with the caveat that this approach
introduces significantly higher latency).  In such a setting, anyone with
access to the logs could verify not only the actions of the mintettes, but
could also replay these actions to compare the ledger agreed upon by the
mintettes and the ledger published by the bank; this would allow an auditor to
ensure that the bank was not engaging in misbehavior by, e.g.,
dropping transactions.

\subsubsection{A simplified block structure}\label{sec:embedding}

The above description of higher-level blocks (and the previous description of
lower-level blocks) contains a number of additional values that do not exist
in the blocks of existing cryptocurrencies, making \rscoin somewhat
incompatible with their semantics.  To demonstrate that \rscoin can more
strongly resemble these cryptocurrencies, we briefly describe a way of
embedding these additional values into the set of transactions.

Rather than include the set $\periodmintette{i+1}$, the bank could instead
store some units of currency in a master address $\addr_\bank$ and include in
$\txset_i$ a transaction $\txspec{\addr_\bank}{n_{\pk}}{\periodpk{i+1}}$,
where $\periodpk{i+1}$ is an address specific to $\period_{i+1}$.
The bank could then include in $\txset_i$ a transaction
$\txspec{\periodpk{i+1}}{n_\mintette}{\pk_\mintette}$
for each mintette $\mintette$ authorized for $\period_{i+1}$.
Now, to check the validity of a particular lower-level block, one could check
that such a transaction was included in the previous higher-level block.

\subsection{Incentivizing mintettes}\label{sec:incentives}

One might naturally imagine that this structure, as currently described,
places the significant burden on the central bank of having to merge the
distinct blocks from each mintette into a consistent history.  By providing
appropriate incentives, however, we can create an environment in which
the presented ledger is in fact consistent before the bank
even sees it.  If mintettes deviate from the expected behavior then, as we
described in Section~\ref{sec:role-bank}, they can be held accountable and
punished accordingly (e.g., not chosen for future periods or not given any
fees they have earned).

Section~\ref{sec:role-bank} describes one direct incentive for mintettes to
collect transactions, which is fees.  
As we described in Section~\ref{sec:role-bank},
mintettes are rewarded only for \emph{active} participation, so that an
authorized mintette needs to engage with the system in order to earn
fees.
Section~\ref{sec:embedding} describes another direct incentive, which is the
authorization of mintettes by the central bank.  For semantic purposes, the
value $n_\mintette$ used to authorize each mintette for the next period
could be arbitrarily small.  As an incentive, however, this value
could be larger to directly compensate the mintettes for their services.
%

Finally, we expect that the central bank could be a national or international
entity that has existing relationships with, e.g., commercial banks.  There
thus
already exist strong business incentives and regulatory frameworks for such
entities to act as honest mintettes.

\subsection{Setting system parameters}\label{sec:parameters}

As described, the system is parameterized by a number of variables, such as
the length of epochs, the length of a period, and the number of
mintettes.
The length of an epoch for an individual mintette is entirely dependent on
the rate at which it processes transactions (as described in detail in
Section~\ref{sec:auditability}).  Mintettes that process more transactions
will therefore have shorter epochs than ones that do so less frequently.
There is no limit on how short an epoch can be, and the only
upper limit is that an epoch cannot last longer than a period.

It might seem desirable for periods to be as
short as possible, as ultimately a transaction is sealed into the official
ledger only at the end of a period.  To ease the burden on the bank, however,
it is also desirable to have longer periods, so that central banks have to
intervene as infrequently as possible (and, as we describe in
Section~\ref{sec:bloat}, so that central banks can potentially perform
certain optimizations to reduce transaction bloat).  In
Section~\ref{sec:basic-consensus}, we described
methods by which mintettes could ``promise'' (in an accountable way) to users
that their transactions would be included, so that in practice
near-instantaneous settlement can be achieved even with longer periods, so
long as one trusts the mintette quorum. 
Bitcoin,
Nevertheless, we do not expect periods to last longer than a day.

For the purposes of having a fair and competitive settlement process, it
is desirable to have as many mintettes as possible; as we saw in
Section~\ref{sec:envelope-numbers}, this is also desirable from a performance
perspective, as the performance of the \rscoin system
(measured in the rate of transactions processed) scales linearly with the
number of mintettes.  Adding more mintettes, however, also has the effect
that they earn less in transaction fees, so these opposing concerns must be
taken into account when settling on a concrete number (to give a very rough
idea, one number that has been suggested~\cite{boe-private} is 200).

\section{Optimizations and Extensions}\label{sec:extensions}

In Sections~\ref{sec:consensus} and~\ref{sec:system}, we presented a
(relatively) minimal version of \rscoin, which allows us to achieve the basic
integrity and scalability properties that are crucial for any currency
designed to be used on a global level.  Here, we briefly sketch some
extensions that could be adopted to strengthen either of these properties,
and leave a more detailed analysis of these or other solutions as interesting
future research.

\subsection{Pruning intermediate transactions}\label{sec:bloat}

At the end of a period, the central bank publishes a higher-level block
containing the collection of transactions that have taken place in that time
interval; it is only at this point that transactions are officially recorded
in the ledger.  Because mintettes provide evidence on a shorter time scale
that a user's transaction is valid and will be included in the ledger,
however,
users might feel more comfortable moving currency multiple times
within a period than in traditional cryptocurrencies (in which one must wait
for one or several blocks to avoid possible double-spending).

It therefore might be the case that at the end of a period, the central bank
sees not just individual transactions, but
potentially multiple ``hops'' or even whole ``chains'' of transactions.  To
limit \emph{transaction bloat}, the bank could thus prune these intermediate
transactions at the end of the period, so that ultimately only the start and
end points of the transaction appear in the ledger, in a new transaction
signed by the central bank.

On its surface, this idea may seem to require a significant amount of trust in
the central bank, as it could now actively modify the transaction history.
The action logs, however, would reveal the changes that the bank had made and
allow users to audit its behavior, but nevertheless the alterations that could
be made would need be significantly restricted.

\subsection{Further incentives for honest behavior}

In addition to the existing incentives for honest behavior outlined in
Sections~\ref{sec:role-bank} and~\ref{sec:incentives}, mintettes could adopt
a sort of proof-of-stake mechanism, in which they escrow some units
of currency with the central bank and are allowed to collate only a
set of transactions whose collective value does not exceed the escrowed value.
If any issue then arises with the transactions produced by the
mintette (e.g., it has accepted double-spending transactions), the central
bank can seize the escrowed value and remove the double-spending transactions,
so the mintette ultimately pays for this misbehavior out of its own pocket
(and maybe even pays additional fines).

This mechanism as described is not fully robust (as in particular the mintette
might accept many expenditures of the same unit of currency, not just
two), but it does have an interesting effect on the length of periods.  In
particular, the length of earlier periods will necessarily
be quite small, as mintettes will not have much capital to post.  As
mintettes accumulate stores of currency, however, periods can grow longer.
This is a fairly natural process, as it also allows for a trial period in the
beginning to ensure that authorized mintettes don't misbehave, and
then for a more stable system as a set of trustworthy mintettes emerges.

\subsection{Multiple banks and foreign exchange}\label{sec:multiple-banks}

In a global setting, one might imagine that each central
bank could develop their own version of \rscoin; this would lead, however, to
a landscape much the same as today's Bitcoin and the many altcoins
it has inspired, in which multiple implementations of a largely overlapping
structure lead to an \emph{infrastructure fragmentation}: bugs are
replicated across codebases and compatibility across different altcoins is
artificially low.

An attractive approach is for different central banks to instead use the same
platform, to prevent this fragmentation and to allow users to seamlessly store
value in many different currencies.  While this allows the currencies
generated
by different central banks to achieve some notion of interoperability, we
still
expect that different blockchains will be kept separate; i.e., a particular
central bank does not\dash and should not\dash have to keep track of
all transactions that are denominated in the currency of another central bank.
(Mintettes, however, may choose to validate transactions for any number of
central banks, depending on their business interests.)

While every central bank does not necessarily
need to be aware of transactions denominated in the currency of another
central bank, this awareness may at times be desirable.  For example, if a
user would like to exchange some units of one currency into another belonging
to a central bank that is relatively known to and trusted by the first (e.g.,
exchange GBP for USD), then this should be a relatively easy process.  The
traditional approach is to simply go to a third-party service that holds units
of both currencies, and then perform one transaction to send units of the
first currency to the service, which will show up in the ledger of the
first currency, and another transaction to receive units of the second
currency, which will show up in the ledger of the second currency.  

Although this is the approach by far most commonly adopted in practice (both
in fiat currency and cryptocurrency markets), it has
a number of limitations, first and foremost of which is that it is completely
opaque: even an outside observer who is able to observe both ledgers sees two
transactions that are not linked in any obvious way.  One might naturally
wonder, then, if a more \emph{transparent} mechansim is possible, in which the
currency exchange shows up as such in the ledger.
We answer this question in the affirmative in the Appendix,
in which we demonstrate a form of \emph{fair exchange}.

Briefly, to achieve this fair exchange, we adapt a protocol to
achieve \emph{atomic cross-chain trading},\footnote{The clearest explanation
    of this for Bitcoin, by Andrew Miller, can be found
    at \url{bitcointalk.org/index.php?topic=193281.msg3315031\#msg3315031}.}
which provides a Bitcoin-compatible
way for two users to \emph{fairly} exchange units of one currency for some
appropriate units of another currency; i.e., to exchange currency in a way
that guarantees that either the exchange is successful or both users end up
with nothing (so in particular it cannot be the case that one user reclaims
currency and the other does not).  If one is less concerned about
compatibility with Bitcoin, then a slightly simpler approach such as ``pay on
reveal secret''~\cite{cross-chain-exchange} could be adopted.

To fit our setting, in which central banks may want to maintain some control
over which other currencies their currency is traded into and out of (and in
what volume), we modify the existing protocol to require a third party to sign
both transactions only if they are denominated in currencies that are viewed
as ``exchangeable'' by that party.  This serves to not only signal the third
party's blessing of the exchange, but also to bind the two transactions
together across their respective blockchains.  Our proposal of this protocol
thus enables transparent exchanges that can be approved by a third party, but
does not (and cannot) prevent exchanges from taking place without this
approval.
Importantly, however, an auditor can now\dash with access to both
blockchains\dash observe the exchange.

\section{Conclusions}

In this paper, we have presented the first cryptocurrency framework, \rscoin,
that provides the control over monetary policy that entities such as central
banks expect to retain.  By constructing a blockchain-based approach that
makes
relatively minimal alterations to the design of successful cryptocurrencies
such
as Bitcoin, we have demonstrated that this centralization can be achieved
while
still maintaining the transparency guarantees that have made
(fully) decentralized cryptocurrencies so attractive.  We have also
proposed a new consensus mechanism based on 2PC and
measured its performance, illustrating that centralization of some authority
allows
for a more scalable system to prevent double spending that completely avoids
the wasteful hashing required in proof-of-work-based systems.
%

\ifproceedings{
\section*{Acknowledgements}

We thank Robleh Ali, Simon Scorer, Alex Mitchell, and John Barrdear from the
Bank of England and Ben Laurie from Google for interesting discussions.  We
also thank our anonymous reviewers and our shepherd, Joseph Bonneau, for their
helpful feedback.  George Danezis is supported in part by EPSRC Grant EP/M013286/1 and H2020 Grant PANORAMIX (ref.\ 653497)
and Sarah Meiklejohn is supported in part by
EPSRC Grant EP/M029026/1.
}\fi

{
\balance
\bibliographystyle{IEEEtranS}
\begin{flushleft}

\end{flushleft}
}

\appendix

\section{Fair Currency Exchange}\label{sec:fair-fx}

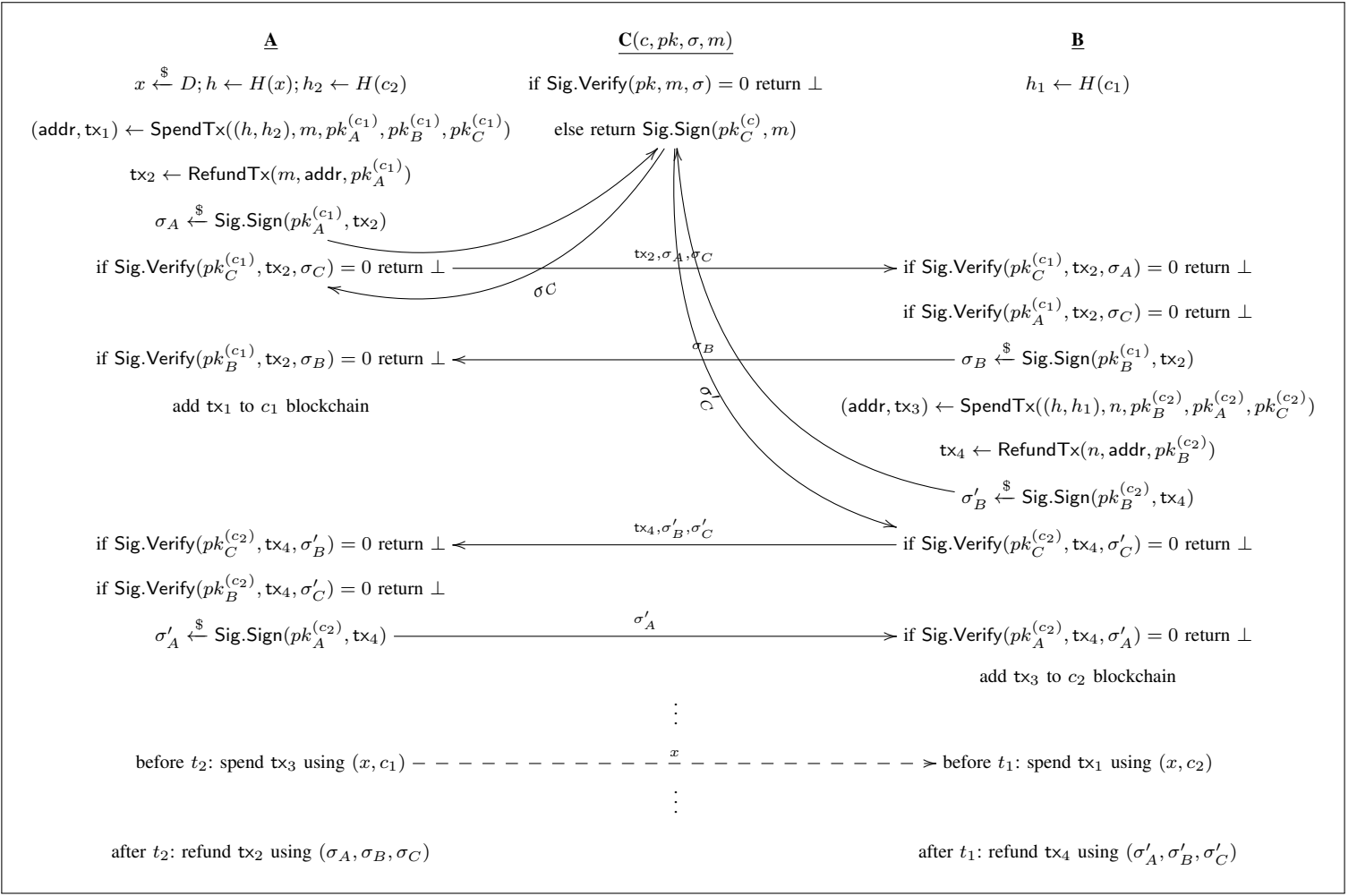
\begin{figure*}[t]
\begin{framed}
\centering
{\footnotesize
\xymatrix@R0.1pc@C0.1pc{
\underline{\tbf{A}} & \underline{\tbf{C}(c,\pk,\sig,m)} &
\underline{\tbf{B}}\\
x\randpick D; h\gets H(x); h_2\gets H(c_2) &
    \trm{if } \sigverify(\pk,m,\sig)=0 \trm{ return } \bot &
    h_1\gets H(c_1)\\
(\addr,\tx_1)\gets\spendtx((h,h_2),m,\currencypk{A}{c_1},\currencypk{B}{c_1},
    \currencypk{C}{c_1}) & \trm{else return } \sigsign(\currencysk{C}{c},m)
    \ar@/^3.2pc/[lddd]^----*---[flip][@]{\sig_C}
    \ar@/_5.7pc/[rddddddddd]^----*---[@]{\sig_C'}&\\
\tx_2\gets\refundtx(m,\addr,\currencypk{A}{c_1}) &&\\
\sig_A\randpick\sigsign(\currencysk{A}{c_1},\tx_2)
    \ar@/_2.5pc/[ruu]&&\\
\trm{if }\sigverify(\currencypk{C}{c_1},\tx_2,\sig_C) = 0 \trm{ return }\bot
    \ar[rr]^-{\tx_2,\sig_A,\sig_C} &&
\trm{if }\sigverify(\currencypk{C}{c_1},\tx_2,\sig_A) = 0\trm{ return } \bot\\
&&\trm{if }\sigverify(\currencypk{A}{c_1},\tx_2,\sig_C) = 0\trm{ return }
\bot\\
\trm{if $\sigverify(\currencypk{B}{c_1},\tx_2,\sig_B) = 0$ return $\bot$} &&
    \ar[ll]_-{\sig_B}\sig_B\randpick\sigsign(\currencysk{B}{c_1},\tx_2)\\
\trm{add $\tx_1$ to $c_1$ blockchain} &&
(\addr,\tx_3)\gets\spendtx((h,h_1),n,\currencypk{B}{c_2},\currencypk{A}{c_2},
    \currencypk{C}{c_2})\\
&& \tx_4\gets\refundtx(n,\addr,\currencypk{B}{c_2})\\
&&\ar@/^4.8pc/[luuuuuuuu]
    \sig_B'\randpick\sigsign(\currencysk{B}{c_2},\tx_4)\\
\trm{if } \sigverify(\currencypk{C}{c_2},\tx_4,\sig_B') = 0\trm{ return }\bot
    && \ar[ll]_-{\tx_4,\sig_B',\sig_C'}
    \trm{if } \sigverify(\currencypk{C}{c_2},\tx_4,\sig_C')=0\trm{ return
}\bot\\
\trm{if } \sigverify(\currencypk{B}{c_2},\tx_4,\sig_C') = 0\trm{ return }\bot
&&\\
\sig_A'\randpick\sigsign(\currencysk{A}{c_2},\tx_4) \ar[rr]^-{\sig_A'} &&
    \trm{if $\sigverify(\currencypk{A}{c_2},\tx_4,\sig_A') = 0$ return
    $\bot$}\\
&& \trm{add $\tx_3$ to $c_2$ blockchain}\\
&\vdots&\\
\trm{before $t_2$: spend $\tx_3$ using $(x,c_1)$} \ar@{-->}[rr]^-{x} &&
    \trm{before $t_1$: spend $\tx_1$ using $(x,c_2)$}\\
& \vdots & \\
\trm{after $t_2$: refund $\tx_2$ using $(\sig_A,\sig_B,\sig_C)$} & &
    \trm{after $t_1$: refund $\tx_4$ using $(\sig_A',\sig_B',\sig_C')$}
}
}
\end{framed}
\caption{A method for $A$ and $B$ to\dash with the approval of a third party
$C$\dash exchange $m$ units of currency $c_1$ for $n$ units of currency $c_2$
in a fair manner; i.e., in a way such that if either $A$ or $B$ stops
participating at any point in the interaction, the other party loses nothing.}
\label{fig:fx}
\end{figure*}

In Section~\ref{sec:multiple-banks}, we described a protocol for atomic
trading of different currencies and outlined some of its features, such as
allowing trade only across authorized currencies (as determined by some third
party).  Our formal protocol that achieves this fair exchange is presented
in Figure~\ref{fig:fx}.

Informally, if Alice and Bob wish to exchange $m$ units of currency $c_1$
for $n$ units of currency $c_2$,
with the blessing of a third party Carol, then they each create two types of
transactions: a ``spend'' transaction, in which the sender releases the units
of currency to one of two addresses, and a ``refund'' transaction, in which
the sender can reclaim the currency after a certain amount of time has passed.
The two addresses in Alice's spend transactions are a
``multi-signature'' address from which funds can be released only with the
signatures of Alice, Bob, and Carol, or Bob's address, from which he can spend
the funds only with knowledge of the pre-image of some hash $H(x)$.  Her
refund transaction then sends the currency back to Alice's address if
signatures are provided by all three parties, and if an
appropriate amount of time $t_1$ has elapsed since the spend transaction was
accepted into the blockchain.  Similarly, Bob's spend transaction requires
Alice to present the pre-image $x$ in order to redeem the funds, and his
refund transaction can be spent only after some time $t_2$ has passed.

\begin{gather*}
\underline{\spendtx(\vec{h},v,\pk_1,\pk_2,\pk_3)}\\
\addr\gets\left\{\begin{matrix}
\multiaddr(\pk_1,\pk_2,\pk_3) & \trm{if $t > t_1$}\\
\pk_2 & \trm{if $H(x_i) = h[i]~\forall i$}
\end{matrix}\right.\\
\trm{return } (\addr,\txspec{\pk_1}{v}{\addr})\\
~\\
\underline{\refundtx(v,\addr_\inputs,\addr_\outputs)}\\
\trm{return } \txspec{\addr_\inputs}{v}{\addr_\outputs} \\
\end{gather*}

Alice begins by creating her spend and refund transactions, as well as picking
the value $x$ and computing $H(x)$.  She then ``commits'' to the currency
$c_2$ being traded with using a second hash $h_2$ and sends the
refund transaction, signed by herself, to Carol.  If Carol is satisfied with
the proposed exchange, she can sign the transaction and give this signature to
Alice.  Alice now solicits a signature from Bob; once she has signatures from
both Bob and Carol, she now has a transaction that she can use to refund her
currency after time $t_1$ has passed.  Thus, it is safe for her to publish the
spend transaction in the blockchain for $c_1$.  Bob then follows suit by
creating his own spend and refund transactions, soliciting signatures from
Alice and Carol, and publishing his spend transaction once he has a valid
refund transaction that he can use if necessary.

Once both transactions are accepted into their respective blockchains,
Alice\dash who so far is the only one with knowledge of the pre-image $x$\dash
can redeem the $n$ units of currency $c_2$ using Bob's spend transaction;
in doing so, she implicitly reveals $x$.  Thus, Bob can now redeem the $m$
units of currency $c_1$ using Alice's spend transaction and the exchange is
complete.  If Alice does not redeem Bob's spend transaction, then after time
$t_2$ Bob can use his refund transaction to redeem the currency himself (so
it is important that $t_2 < t_1$).

\end{document}